\documentclass[a4paper,11pt]{amsart}
\usepackage{amssymb}
\usepackage{tikz}\usetikzlibrary {decorations.pathreplacing,decorations.markings,patterns,shapes,backgrounds}
\usepackage{hyperref}
\usepackage[a4paper,left=2.25cm,right=2.25cm,top=3cm,bottom=3cm]{geometry}
\usepackage{pdflscape}

\newcommand{\braket}[2]{\left< #1 \vphantom{#2}\right|
\! \left. #2 \vphantom{#1} \right>}%

\title{The trigonometric $E_8$ $R$-matrix}
\author{Paul Zinn-Justin}
\address{Paul Zinn-Justin, School of Mathematics and Statistics, The University of Melbourne, 
Victoria 3010, Australia}
\email{pzinn@unimelb.edu.au}
\thanks{PZJ was supported by ARC grants FT150100232 and DP180100860.
He would like to thank A.~Kuniba for useful discussions and J.~Lamers for comments on the manuscript.
}
\newtheorem{lem}{Lemma}
\newcommand\rem[2][]{}
\newcommand\CC{{\mathbb C}}
\newcommand\RR{{\mathbb R}}
\newcommand\ZZ{{\mathbb Z}}
\newcommand\Uq{{\mathcal U}_q}
\newcommand\e{{\mathfrak e_8}}
\newcommand\ee{{\widehat{\mathfrak e_8}}}
\newcommand\cf[1]{\mathrm{cf}\!\left[#1\right]}
\newcommand\cR{{\check R}}
\newcommand\height{\mathrm{ht}}
\tikzset{Rmat/.style={circle,fill=black,inner sep=0.15cm}}
\tikzset{empty/.style={circle,solid,thin,sharp corners,draw=black,inner sep=1.5pt,fill=white}}
\tikzset{vertex/.style={circle,solid,thin,sharp corners,draw=black,inner sep=1.5pt,fill=red!50!blue}}
\tikzset{dvertex/.style={ellipse,solid,thin,sharp corners,draw=black,inner xsep=3pt,inner ysep=1.5pt,fill=red!50!blue}}
\makeatletter
\tikzset{circle split part fill/.style  args={#1,#2}{%
 alias=tmp@name,
  postaction={%
    insert path={
     \pgfextra{%
     \pgfpointdiff{\pgfpointanchor{\pgf@node@name}{center}}%
                  {\pgfpointanchor{\pgf@node@name}{east}}%
     \pgfmathsetmacro\insiderad{\pgf@x}
      \fill[#1] (\pgf@node@name.base) ([xshift=-\pgflinewidth]\pgf@node@name.east) arc
                          (0:180:\insiderad-\pgflinewidth)--cycle;
      \fill[#2] (\pgf@node@name.base) ([xshift=\pgflinewidth]\pgf@node@name.west)  arc
                           (180:360:\insiderad-\pgflinewidth)--cycle;                    }}}}}  
 \makeatother
\tikzset{halfvertex/.style={circle,solid,thin,sharp corners,draw=black,inner sep=1pt,circle split,circle split part fill={white,red!50!blue}}}%
\def\bb{}
\definecolor{bgcolor}{rgb}{1,0.97,1}
\tikzset{every picture/.style={inner frame sep=3pt,background rectangle/.style={rounded corners,fill=bgcolor},show background rectangle}}
\newcommand\crossing[1][1]{\tikz[baseline=-3pt,scale=#1,inner frame sep=#1*3pt]{\bb
\draw (-1,1) node[vertex] {} -- (-0.1,0.1) (0.1,-0.1) -- (1,-1)  node[vertex] {};
\draw (-1,-1) node[vertex] {} -- (1,1)  node[vertex] {};
}}
\newcommand\invcrossing[1][1]{\tikz[baseline=-3pt,scale=#1,inner frame sep=#1*3pt]{\bb
\draw (1,1) node[vertex] {} -- (0.1,0.1) (-0.1,-0.1) -- (-1,-1)  node[vertex] {};
\draw (1,-1) node[vertex] {} -- (-1,1)  node[vertex] {};
}}
\newcommand\flatcrossing[1][1]{\tikz[baseline=-3pt,scale=#1,inner frame sep=#1*3pt]{\bb
\draw (1,1) node[vertex] {} -- (-1,-1)  node[vertex] {};
\draw (1,-1) node[vertex] {} -- (-1,1)  node[vertex] {};
}}
\newcommand\tri[1][1]{\tikz[baseline=1cm]{\draw[bend left=60] (1,1) node[vertex] {} to (-1,1) node[vertex] {}; \draw (0,1) node[vertex] {} -- (0,0.5);}}
\newcommand\triinv[1][1]{\tikz[baseline=0cm,scale=-#1]{\draw[bend left=60] (1,1) node[vertex] {} to (-1,1) node[vertex] {}; \draw (0,1) node[vertex] {} -- (0,0.5);}}
\newcommand\mycap[1][1]{\tikz[baseline=-3pt]{\draw[bend left=60] (-1,0) node[vertex] {} to (1,0) node[vertex] {};}}
\newcommand\mycup[1][1]{\tikz[baseline=-3pt]{\draw[bend left=60] (1,0) node[vertex] {} to (-1,0) node[vertex] {};}}
\newcommand\D[1][1]{\tikz[baseline=0cm,scale=#1,inner frame sep=#1*3pt]{
\draw (-1,-1) node[vertex] {} -- (0,-0.2) -- (1,-1) node[vertex] {};
\draw (0,-0.2) -- (0,1) node[vertex] {};
}}
\newcommand\U[1][1]{\tikz[baseline=0cm,scale=#1,inner frame sep=#1*3pt]{
\draw (0,-1) node[vertex] {} -- (0,0.2);
\draw (-1,1) node[vertex] {} -- (0,0.2) -- (1,1) node[vertex] {};
}}
\newcommand\full[1][1]{\tikz[baseline=-3pt,scale=#1,inner frame sep=#1*3pt]{\bb\draw[bend right=45] (-1,-1) node[vertex] {} to (-1,1) node[vertex] {} (1,1) node[vertex] {} to (1,-1) node[vertex] {};}}
\newcommand\identity[1][1]{\tikz[baseline=-3pt,scale=#1,inner frame sep=#1*3pt]{\bb\draw[bend right=45,dashed] (-1,-1) node[halfvertex] {} to (-1,1) node[halfvertex] {} (1,1) node[halfvertex] {} to (1,-1) node[halfvertex] {};}}
\newcommand\cupcap[1][1]{\tikz[baseline=-3pt,scale=#1,inner frame sep=#1*3pt]{\bb\draw[bend left=45] (-1,-1) node[vertex] {} to (1,-1) node[vertex] {} (1,1) node[vertex] {} to (-1,1) node[vertex] {} ;}}
\newcommand\vertical[1][1]{\tikz[baseline=-3pt,scale=#1,inner frame sep=#1*3pt]{\bb
\draw (-1,-1) node[vertex] {} -- (0,-0.5) -- (1,-1) node[vertex] {} ;
\draw (-1,1) node[vertex] {} -- (0,0.5) -- (1,1) node[vertex] {};
\draw (0,-0.5) -- (0,0.5);
}}
\newcommand\horizontal[1][1]{\tikz[baseline=-3pt,scale=#1,inner frame sep=#1*3pt,rotate=90]{\bb
\draw (-1,-1) node[vertex] {} -- (0,-0.5) -- (1,-1) node[vertex] {} ;
\draw (-1,1) node[vertex] {} -- (0,0.5) -- (1,1) node[vertex] {} ;
\draw (0,-0.5) -- (0,0.5);
}}
\newcommand\mysquare[1][1]{\tikz[baseline=-3pt,scale=#1,inner frame sep=#1*3pt]{\bb
\draw (-1,-1) node[vertex] {} -- (-0.5,-0.5) -- (-0.5,0.5) -- (-1,1) node[vertex] {};
\draw (1,-1) node[vertex] {} -- (0.5,-0.5) -- (0.5,0.5) -- (1,1) node[vertex] {};
\draw (-0.5,-0.5) -- (0.5,-0.5) (-0.5,0.5) -- (0.5,0.5);
}}
\newcommand\doublesquare[1][1]{\tikz[baseline=-3pt,scale=#1,inner frame sep=#1*3pt]{\bb
\draw (-1,-1) node[vertex] {} -- (-0.5,-0.5) -- (0.5,-0.5) -- (1,-1) node[vertex] {};
\draw (-1,1) node[vertex] {} -- (-0.5,0.5) -- (0.5,0.5) -- (1,1) node[vertex] {};
\draw (-0.5,-0.5) -- (-0.25,0) -- (-0.5,0.5) (0.5,-0.5) -- (0.25,0) -- (0.5,0.5) (-0.25,0) -- (0.25,0);
}}
\newcommand\leftid[1][1]{\tikz[baseline=-3pt,scale=#1,inner frame sep=#1*3pt]{\bb
\draw[bend right=45] (-1,-1) node[vertex] {} to (-1,1) node[vertex] {} (1,1);
\node[empty] at (1,-1) {};
\node[empty] at (1,1) {};
}}
\newcommand\diagida[1][1]{\tikz[baseline=-3pt,scale=#1,inner frame sep=#1*3pt]{\bb
\draw (-1,1) node[vertex] {} to (1,-1) node[vertex] {};
\node[empty] at (-1,-1) {};
\node[empty] at (1,1) {};
}}
\newcommand\diagidb[1][1]{\tikz[baseline=-3pt,scale=#1,inner frame sep=#1*3pt]{\bb
\draw (1,1) node[vertex] {} to (-1,-1) node[vertex] {};
\node[empty] at (1,-1) {};
\node[empty] at (-1,1) {};
}}
\newcommand\rightid[1][1]{\tikz[baseline=-3pt,scale=#1,inner frame sep=#1*3pt]{\bb
\draw[bend right=45] (1,1) node[vertex] {} to (1,-1) node[vertex] {};
\node[empty] at (-1,-1) {};
\node[empty] at (-1,1) {};
}}
\newcommand\Yaa[1][1]{\tikz[baseline=-3pt,scale=#1,inner frame sep=#1*3pt]{\bb
\draw (-1,1) node[vertex] {} to (1,-1) node[vertex] {};
\node[empty] at (-1,-1) {};
\draw (1,1) node[vertex] {} -- (0,0);
}}
\newcommand\Yba[1][1]{\tikz[baseline=-3pt,scale=#1,inner frame sep=#1*3pt]{\bb
\draw (1,1) node[vertex] {} to (-1,-1) node[vertex] {};
\node[empty] at (1,-1) {};
\draw (-1,1) node[vertex] {} -- (0,0);
}}
\newcommand\Yab[1][1]{\tikz[baseline=-3pt,scale=#1,inner frame sep=#1*3pt]{\bb
\draw (-1,1) node[vertex] {} to (1,-1) node[vertex] {};
\node[empty] at (1,1) {};
\draw (-1,-1) node[vertex] {} -- (0,0);
}}
\newcommand\Ybb[1][1]{\tikz[baseline=-3pt,scale=#1,inner frame sep=#1*3pt]{\bb
\draw (1,1) node[vertex] {} to (-1,-1) node[vertex] {};
\node[empty] at (-1,1) {};
\draw (1,-1) node[vertex] {} -- (0,0);
}}
\newcommand\capa[1][1]{\tikz[baseline=-3pt,scale=#1,inner frame sep=#1*3pt]{\bb
\draw[bend left=45] (-1,-1) node[vertex] {} to (1,-1) node[vertex] {};
\node[empty] at (-1,1) {};
\node[empty] at (1,1) {};
}}
\newcommand\cupa[1][1]{\tikz[baseline=-3pt,scale=#1,inner frame sep=#1*3pt]{\bb
\draw[bend left=45] (1,1) node[vertex] {} to (-1,1) node[vertex] {} ;
\node[empty] at (-1,-1) {};
\node[empty] at (1,-1) {};
}}
\newcommand\myempty[1][1]{\tikz[baseline=-3pt,scale=#1,inner frame sep=#1*3pt]{\bb
\node[empty] at (-1,1) {};
\node[empty] at (1,1) {};
\node[empty] at (-1,-1) {};
\node[empty] at (1,-1) {};
}}
\newcommand\dbl[1][1]{\tikz[baseline=0cm,scale=#1,inner frame sep=#1*3pt]{
\draw (-1,-1) node[vertex] {} -- (0,-0.2) -- (1,-1) node[vertex] {};
\draw[double distance=4pt,double=bgcolor,black] (0,1) node[dvertex] {} -- (0,-0.3);
}}
\begin{document}
\begin{abstract}
An expression for the $R$-matrix associated to $\Uq(\ee)$ in its $249$-dimensional representation
is given using the diagrammatic calculus of $\Uq(\e)$ invariants.
\end{abstract}

\maketitle

\section{Introduction}
Quantized affine algebras are pseudo-triangular Hopf algebras \cite{Drinfeld-QG}; 
as such, given any pair
of ``generic'' representations $V_1$ and $V_2$, there exists an intertwiner between $V_1\otimes V_2$ and $V_2\otimes V_1$,
the so-called $R$-matrix. $R$-matrices are the central object in quantum integrable systems,
and for most algebras and many low-dimensional representations, they are known explicitly
\cite{KRS-YBE,ZamFat,IK-SM,Bazh-Lie,Jimbo-Rmat,Jimbo-YBE,Ma-YBE,Kuniba-G2,DGZ-YBE,Manga-R}. 
As far as the author knows, the only
case that has not been studied in detail is that of (the quantized affine algebra of) $\e$. 
The goal of this short paper is therefore
rather modest: it is the explicit computation of the $R$-matrix associated
to the quantized loop algebra $\Uq(\ee)$ in its lowest nontrivial representation, which is of dimension
$249$ (and is {\em not}\/ irreducible for the nonaffine algebra $\Uq(\e)$, which is a source of complication). 
In order to to do so, we develop a natural diagrammatic language for the theory of invariants of $\Uq(\e)$
based on tensors of the adjoint representation, which we then apply to the computation of the $R$-matrix.
The latter is fixed (up to normalization) 
by solving the equations that require it to be an intertwiner \cite{Jimbo-qg-YBE}.

The main challenge of this paper is computational: for instance, the $R$-matrix is a matrix
of size $249^2=62001$, so altogether an array of approximately 4 billion entries (albeit a sparse array, since only
$0.05\%$ of these are nonzero). This means that many computations in this paper cannot be performed by hand,
and the help of a symbol computation program is necessary. (The author's code is available on request.)
In particular, quite a few results are presented below without proof; if so, the reader should assume that
they are the result of a computer-assisted calculation.

The immediate motivation for this paper came from the work of the author in collaboration with
A.~Knutson \cite{artic71,SchubII} where the $R$-matrix of $\Uq(\ee)$ appeared unexpectedly
in the computation of structure constants of the $K$-theory of 4-step flag varieties.
Another possible source of interest is the fact that the scaling limit of the Ising model at the 
critical temperature in a magnetic field
is known to be related to an $\e$ integrable field theory \cite{Zamo-e8,BNW-e8}, and a vertex lattice model
based on $\e$ might be desirable.

The paper is organized as follows. \S\ref{sec:prelim} contains basic definitions related
to the Hopf algebra $\Uq(\ee)$. \S\ref{sec:nonaff} presents the diagrammatic calculus for $\Uq(\e)$
invariants. $\S\ref{sec:aff}$ is the core of the paper: the $249$-dimensional
representation of $\Uq(\ee)$ is defined, the
trigonometric $R$-matrix introduced and its main properties discussed. 
Finally, appendix \ref{app:mult}
contain some cumbersome diagrammatic identities, 
appendix \ref{app:main} is the full expression of the $R$-matrix,
and appendix \ref{app:rat} discusses the rational limit
of the $R$-matrix.

Intentionally, this paper has very concrete formulae which we hope will be useful in future
studies of the quantum integrable systems associated to $\e$.

\vfill\eject

\section{Preliminaries}\label{sec:prelim}
\subsection{Dynkin diagrams}\label{sec:dynkin}
The Dynkin diagram of $\ee$ is
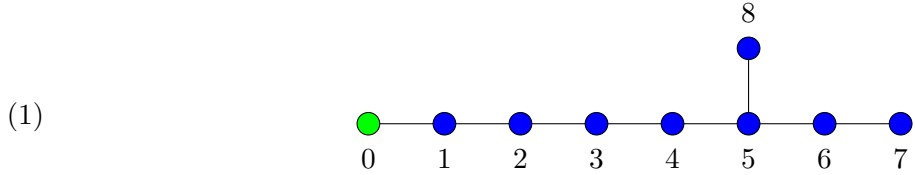
\begin{equation}\label{eq:dynkin}
\tikz[baseline=0,every node/.style={circle,fill=blue,draw=black,inner sep=3pt},background rectangle/.style={}]{
\draw (0,0) -- (7,0) (5,0) -- (5,1) node[label={above:$8$}] {};
\foreach\i in {1,...,7} \node[label={below:$\i$}] at (\i,0) {};
\node[fill=green,label={below:$0$}] at (0,0) {};
}
\end{equation}
Its adjacency matrix $A$ uses the index set
$I=\{0,1,\ldots,8\}$ as on the figure. Let $C=2-A$ be the corresponding Cartan matrix.

If one removes the (green) vertex labeled $0$ in \eqref{eq:dynkin}, one obtains the Dynkin
diagram of $\e$.

\subsection{The $\e$ root system}\label{sec:rootsys}
Let $\Phi=\Phi_+\sqcup \Phi_-$ be the root system of $\e$, divided into
positive and negative roots; $|\Phi|=240$.
Let $Q$ be the root lattice (generated by $\Phi$). 
Let $\alpha_1,\ldots,\alpha_8\in \Phi_+$ be the set of simple roots.

There is a nondegenerate scalar product $\braket{\cdot}{\cdot}$
on $Q\otimes_{\ZZ} \RR \cong \RR^8$ (inverse of the Killing form on the Cartan subalgebra of $\e$), 
such that $\Phi$ is exactly the subset
of $Q$ of vectors of square length $2$. More generally, one has, for $\beta,\gamma\in\Phi$,
$\braket{\beta}{\gamma}\in \{-2,-1,0,+1,+2\}$ with $\braket{\beta}{\gamma}=\pm2$
iff $\beta=\pm\gamma$. 

The dual basis of the simple roots $\alpha_i$ is that of fundamental weights $\omega_i$:
$\braket{\alpha_i}{\omega_j}=\delta_{i,j}$.

Define the sequence of integers
\begin{equation}\label{eq:n}
(n_0,n_1,\ldots,n_8)=(1,2,3,4,5,6,4,2,3)
\end{equation}
The highest root of $\e$ (highest weight of the adjoint representation,
and therefore a fundamental weight) is given by $\omega_1=\sum_{i=1}^8 n_i\alpha_i$. 
Alternatively, define $\alpha_0=-\omega_1$ to be the lowest root of $\e$; it then satisfies
\[
\sum_{i=0}^8 n_i \alpha_i=0
\]
The Cartan matrix of \S\ref{sec:dynkin} encodes the scalar products: $C_{ij}=\braket{\alpha_i}{\alpha_j}$
for $i,j\in I$ (including the value $0$).

We need three more definitions.

Given $\beta\in\Phi$, define $\height(\beta)$ to be the sum of its entries in the basis of simple roots,
that is, $\height(\beta)=\braket{\rho}{\beta}$ where $\rho=\sum_{i=1}^8 \omega_i=\frac{1}{2}\sum_{\beta\in \Phi_+}\beta$.

Denote $\Phi^3_0$ the set of triples $(\beta,\gamma,-\beta-\gamma)\in\Phi^3$ that
sum to zero.
Define the binary relation $\to$ on $\Phi^3_0$ by
$(\beta,\gamma,-\beta-\gamma)\to (\beta+\alpha_i,\gamma-\alpha_i,-\beta-\gamma)$,
$(\beta,\gamma,-\beta-\gamma)\to (\beta,\gamma+\alpha_i,-\beta-\gamma-\alpha_i)$ and
$(\beta,\gamma,-\beta-\gamma)\to (\beta+\alpha_i,\gamma,-\beta-\gamma-\alpha_i)$
for $i=1,\ldots,8$. Given $x\in\Phi^3_0$, write $\ell(x)$ for the maximum of lengths $k$ of chains
$x\to x_1 \to\cdots\to x_{k}$ starting from $x$.
\rem[gray]{is there a better definition?}

Finally, define a map $\epsilon: \Phi^2 \to \{-1,+1\}$ by
\[
\epsilon(\beta,\gamma)=\prod_{\substack{1\le i<j\le 8\\[1pt]
C_{ij}
=-1}} (-1)^{\braket{\alpha_i}{\beta}\braket{\alpha_j}{\gamma}}
\]
Note that it satisfies
\[
\epsilon(\beta,\gamma)\epsilon(\gamma,\beta)=(-1)^{\braket{\beta}{\gamma}}
\]

\subsection{The quantized affine algebra $\Uq(\ee)$}
In all that follows, $q$ is a ``generic'' nonzero complex number, i.e., not a root of unity.
 
Introduce the notation for $q$-numbers
\[
[n]=\frac{q^{n}-q^{-n}}{q-q^{-1}}
\]

The quantized affine algebra $\Uq(\ee)$ is given by
generators $\{e_i,f_i,k_i^{\pm},\ i\in I\}$ and relations (see e.g.~\cite[Chap.~12]{CP-book})
\begin{align}\notag
&k_i k_j=k_jk_i
&&
[e_i,f_j] = \delta_{ij} \frac{k_i-k_i^{-1}}{q-q^{-1}}
\\
&k_i e_j k_i^{-1} = q^{C_{ij}} e_j
&&
k_i f_j k_i^{-1} = q^{-C_{ij}} f_j
\\\notag
&e_ie_j=e_je_i&& f_if_j=f_jf_i&&\text{if }C_{ij}=0
\\\notag
&e_ie_j^2-[2]e_je_ie_j+e_j^2e_i=0
&&
f_if_j^2-[2]f_jf_if_j+f_j^2f_i=0
&&
\text{if }C_{ij}=-1
\end{align}
for all $i,j\in I$.

The coproduct, antipode and counit are defined by
\begin{align}\notag\label{eq:co}
\Delta(e_i) &= e_i \otimes 1 + k_i \otimes e_i
&
S(e_i)&=-k_i^{-1}e_i
&
\varepsilon(e_i)&=0
\\
\Delta(f_i) &= f_i \otimes k_i^{-1} + 1 \otimes f_i
&
S(f_i)&=-f_ik_i
&
\varepsilon(f_i)&=0
\\\notag
\Delta(k_i) &= k_i \otimes k_i
&
S(k_i)&=k_i^{-1}
&
\varepsilon(k_i)&=1
\end{align}

The Cartan matrix being degenerate has two consequences. Firstly,
there is a central element $\prod_{i=0}^8 k_i^{n_i}$ where the integers $n_i$ were defined in \eqref{eq:n}.
In all representations that we consider, this element is sent to $1$.

Secondly, there is a nontrivial
gradation on $\Uq(\ee)$: we choose
$e_0$ to have degree $+1$, $f_0$ degree $-1$ and all other generators degree $0$.

If one restricts the index set of the generators to
$\{1,\ldots,8\}$, one obtains the quantized algebra
$\Uq(\e)$; it is naturally a (horizontal, or degree $0$) sub-Hopf algebra of $\Uq(\ee)$.

\section{\texorpdfstring{$\Uq(\e)$}{Uq(e8)} invariants}\label{sec:nonaff}
We first study the non-affine algebra $\Uq(\e)$. We shall be interested in two representations of it: the trivial representation,
which is simply given by the counit in \eqref{eq:co}, and the 248-dimensional representation (a deformation of the adjoint representation
of $\e$), which we discuss now.

\subsection{The $248$-dimensional representation}\label{sec:248}
The definition is a straightforward $q$-deformation of the adjoint action of $\e$.
$V:=V_{\omega_1}$ is a $248$-dimensional vector space with basis $v_\beta$, $\beta\in\Phi$ and $u_j$, $j=1,\ldots,8$.
The action of $\Uq(\e)$ on it is given by
\begin{align*}
k_i v_\beta &= q^{\braket{\alpha_i}{\beta}} v_\beta
&
k_i u_j &= u_j
\\
e_i v_\beta &= \begin{cases}
\epsilon(\alpha_i,\beta) v_{\beta+\alpha_i}&\beta+\alpha_i\in\Phi
\\
u_i&\beta+\alpha_i=0
\\
0&\text{else}
\end{cases}
&
e_i u_j &= -[C_{ij}] v_{\alpha_i} 
\\
f_i v_\beta &= \begin{cases}
\epsilon(-\alpha_i,\beta) v_{\beta-\alpha_i}&\beta-\alpha_i\in\Phi
\\
-u_i&\beta-\alpha_i=0
\\
0&\text{else}
\end{cases}
&
f_i u_j &= [C_{ij}] v_{-\alpha_i} 
\end{align*}
for $i=1,\ldots,8$.

Because the representation ring of $\Uq(\e)$ (for generic $q$) is isomorphic to
that of $\e$, various results are known. For example, the adjoint representation
of $\e$ has an invariant bilinear form (the Killing form) and an invariant trilinear
form (built out of the bracket). The same is true of $V$:

\begin{lem}
The following bilinear form $B$
is an $\Uq(\e)$ invariant:
\begin{align*}
B(v_\beta\otimes v_{-\beta})&=\epsilon(\beta,\beta)q^{1-\height(\beta)}
\\
B(u_i\otimes u_j)&=[C_{ij}]
\end{align*}
all other entries being zero (as implied by their nonzero weight).
\end{lem}
\begin{proof}
The $k_i$ invariance is trivial. Let us check the $e_i$ invariance (the $f_i$ invariance works similarly, and is in fact implied by $k_i,e_i$
invariance in finite dimension). There are three nontrivial cases to consider: either one starts with $v_\beta\otimes v_\gamma$,
in which case for its image under $Be_i$ to be nonzero one must have $\gamma=-\beta-\alpha_i$; then
\begin{align*}
B(e_i (v_\beta\otimes v_{-\beta-\alpha_i}))&=
B(\epsilon(\alpha_i,\beta)v_{\beta+\alpha_i}\otimes v_{-\beta-\alpha_i}
+q^{\braket{\alpha_i}{\beta}}\epsilon(\alpha_i,-\beta-\alpha_i)v_\beta\otimes v_{-\beta})
\\
&=\epsilon(\alpha_i,\beta)\epsilon(\beta+\alpha_i,\beta+\alpha_i)q^{1-\height(\beta+\alpha_i)}+q^{\braket{\alpha_i}{\beta}}\epsilon(\alpha_i,-\beta-\alpha_i)\epsilon(\beta,\beta)q^{1-\height(\beta)}
\\
&=\epsilon(\alpha_i,\beta)\epsilon(\beta,\beta)q^{-\height(\beta)}(1-q^{1+\braket{\alpha_i}{\beta}})
\\
&=\epsilon(\alpha_i,\beta)\epsilon(\beta,\beta)q^{-\height(\beta)}(1-q^{1+\frac{1}{2}(|\beta+\alpha_i|^2-|\beta|^2-|\alpha_i|^2)})
\\
&=0
\end{align*}
where in the last line we used the fact that all roots of $\Phi$ have square length $2$.

Similarly, the other two cases are
\begin{align*}
B(e_i(u_j\otimes v_{-\alpha_i}))
&=
B(-[C_{ij} v_{\alpha_i}\otimes v_{-\alpha_i}) + u_j\otimes u_i)
\\
&=
-[C_{ij}]q^{1-\height(\alpha_i)} + [C_{ij}]
\\
&=0
\end{align*}
and
\begin{align*}
B(e_i(v_{-\alpha_i}\otimes u_j ))
&=
B(u_i\otimes u_j - q^{-\braket{\alpha_i}{\alpha_i}} v_{-\alpha_i}\otimes [C_{ij}] v_{\alpha_i})
\\
&=[C_{ij}]-q^{-2} q^{1-\height(-\alpha_i)} [C_{ij}]
\\
&=0
\end{align*}


\end{proof}

The trilinear form is more complicated to define:
\begin{lem}
Define $T\in V^*\otimes V^*\otimes V^*$ as follows.

Given a triple $(\beta,\gamma,-\beta-\gamma)\in\Phi^3_0$, write
\[
T(v_\beta\otimes v_\gamma\otimes v_{-\beta-\gamma})=\epsilon(\beta,\beta)\epsilon(\gamma,\gamma)\epsilon(\gamma,\beta) q^{-41+\ell(\beta,\gamma,-\beta-\gamma)}\qquad \beta\in\Phi_+,\ -\beta-\gamma\in\Phi_-
\]
and then use
\[
T(v_\beta\otimes v_\gamma\otimes v_\alpha)=T(v_\alpha\otimes v_\beta\otimes v_\gamma)q^{2\height(\alpha)}
\]
to reduce to the case above.

Also write
\[
\left.
\begin{aligned}
T(u_i\otimes v_{\alpha_i}\otimes v_{-\alpha_i})&=q
\\
T(u_i\otimes v_{-\alpha_i}\otimes v_{\alpha_i})&=-q
\\
T(v_{\alpha_i}\otimes u_i\otimes v_{-\alpha_i})&=-q^{-1}
\\
T(v_{\alpha_i}\otimes v_{-\alpha_i}\otimes u_i)&=q
\\
T(v_{-\alpha_i}\otimes u_i\otimes v_{\alpha_i})&=q^3
\\
T(v_{-\alpha_i}\otimes v_{\alpha_i}\otimes u_i)&=-q
\end{aligned}
\right\}
(q^{d(i)}+q^{-d(i)})
\]
and
\[
T(u_i\otimes u_j\otimes u_k)
=(q-q^{-1})
\begin{cases}
[2]\left(q^{d(i)}+q^{-d(i)}\right)
&i=j=k
\\
-[d(k)]&i=j\leftarrow k\text{ or permutations}
\\
[d(k)]&i=j\rightarrow k\text{ or permutations}
\end{cases}
\]
where $i,j,k$ are viewed as nodes of the Dynkin diagram,
$d(i)$ is the distance to the central node $5$,
and the Dynkin diagram is oriented away from the central node.

All other entries of $T$ are zero.

Then $T$ is an $\Uq(\e)$ invariant.
\end{lem}
The proof is similar to that of the previous lemma, and is skipped.

\medskip

More generally, the 
decomposition of $V^{\otimes 2}$ is known: one has
\begin{equation}\label{eq:dec}
V_{\omega_1}\otimes V_{\omega_1}
\cong \bigoplus_{\omega\in \Omega} V_\omega\qquad
\Omega=\{0,\omega_1,\omega_7,2\omega_1,\omega_2\}
\end{equation}
where $V_\omega$ is the irreducible representation of $\Uq(\e)$ of highest weight $\omega$; 
it is also the unique irreducible representation of dimension $1,248,3875,27000,30380$ 
for $\omega=0,\omega_1,\omega_7,2\omega_1,\omega_2$, respectively.

This implies that $\text{End}(V^{\otimes 2})^{\Uq(\e)}$, 
the space of $\Uq(\e)$ intertwiners of endormorphisms of $V^{\otimes 2}$,
is of dimension $5$. Out of the invariants above we can build a basis for
this space. It is convenient to do so using a diagrammatic language.

\subsection{Diagrammatic depiction of invariants}\label{sec:diag}
We use here the graphical calculus which is standard in mathematical physics, in the particular form
which is adapted to quantum groups \cite{Kassel-QG}, namely our diagrams are {\em planar}.
See also Cvitanovic \cite{Cvita-book} for a discussion of the non $q$-deformed case (for all
simple Lie algebras), so without the planarity requirement.

In general, to a graph embedded in the plane of the form
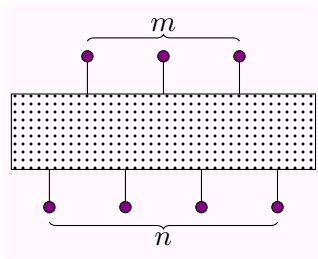
\[
\tikz{\bb
\foreach\i in {1,...,3} \draw (\i+0.5,1) node[vertex] {} -- ++(0,-0.5);
\draw[decoration=brace,decorate] (1.5,1.2) -- node[above] {$m$} (3.5,1.2);
\foreach\i in {1,...,4} \draw (\i,-1) node[vertex] {} -- ++(0,0.5);
\draw[decoration=brace,decorate] (4,-1.2) -- node[below] {$n$} (1,-1.2);
\draw[pattern=dots] (0.5,0.5) rectangle (4.5,-0.5);
}
\]
we will associate an $\Uq(\e)$ intertwiner from $V^{\otimes m}$ to $V^{\otimes n}$,
where tensor products are read from left to right.
As usual, vertical concatenation corresponds to multiplication, whereas horizontal
juxtaposition corresponds to tensor product.

The simplest graph is the identity from $V$ to $V$, which is 
\[
\tikz[baseline=-3pt]{\draw (0,1) node[vertex] {} -- (0,-1) node[vertex] {};}
\]

Next is the bilinear form $B$, which we draw as a cup:
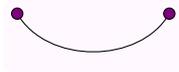
\[
\mycup
\]
Because it is nondegenerate, it possesses an inverse, drawn as a cap:
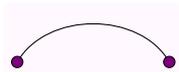
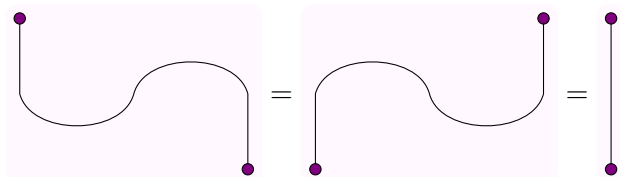
\[
\mycap
\]
with the identities
\[
\tikz[baseline=-3pt]{
\draw (-0.5,1) node[vertex] {} -- (-0.5,0) [bend right=75] to (1,0) [bend left=75] to (2.5,0) -- (2.5,-1) node[vertex] {};
}
=
\tikz[baseline=-3pt,xscale=-1]{
\draw (-0.5,1) node[vertex] {} -- (-0.5,0) [bend right=75] to (1,0) [bend left=75] to (2.5,0) -- (2.5,-1) node[vertex] {};
}
=
\tikz[baseline=-3pt]{\draw (0,1) node[vertex] {} -- (0,-1) node[vertex] {};}
\]

Gluing together cup and cap in the opposite way, we obtain our first
nontrivial identity:
\begin{equation}\label{eq:rela}
\tikz[baseline=-3pt]{\draw (0,0) circle[radius=0.75];} = \frac{[20][24][31]}{[6][10]}
\end{equation}
(we shall not prove this identity, or any that follow, but mention in passing that this is simply
the principal specialization of the character of the adjoint representation of $\e$).

Similarly, the trilinear form $T$ can be drawn as
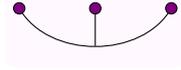
\[
\tri
\]
Combining it with caps, we can get three more intertwiners, namely
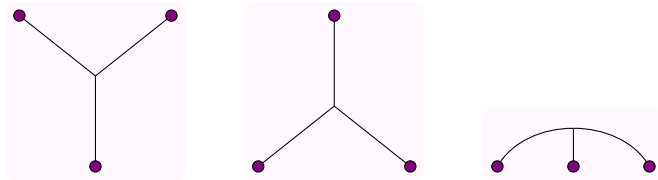
\[
\U\qquad\D\qquad\triinv
\]
In principle, there are several ways to glue caps to $T$ to produce these diagrams,
but by Schur's lemma and normalization condition these all produce the same result.
\rem[gray]{is there a better to say this?}

One also has the trivial identities:
\[
\tikz[baseline=-3pt]{
\draw (0,0) node[vertex] {} -- (0,0.5); \draw (0,1) circle[radius=0.5];
}
=
\tikz[baseline=-3pt]{
\draw (0,0) node[vertex] {} -- (0,-0.5); \draw (0,-1) circle[radius=0.5];
}
=0
\]
as well as two nontrivial ones:\footnote{Note that these identities {\em do}\/ depend
on the choice of normalization of $B$ and $T$, which is determined by imposing that it agree with
the standard $\e$ convention at $q=1$, that it 
be the simplest possible (entries are coprime polynomials)
and that constants occurring in the identities be palyndromic in $q$
(which fixes the remaining freedom of multiplying by a power of $q$).}
\begin{gather}\label{eq:relb}
\tikz[baseline=-3pt]{\draw (0,-1) node[vertex] {} -- (0,-0.5) (0,0.5) -- (0,1) node[vertex] {};
\draw (0,0) circle[radius=0.5];}
=
\frac{[10][15]^2[18][32]}{[5][9][16][30]}
\ \tikz[baseline=-3pt]{\draw (0,-1) node[vertex] {} -- (0,1) node[vertex] {};}
\\
\label{eq:relc}
\tikz[baseline=-3pt]{
\draw (-1.2,1) node[vertex] {} -- (-0.5,0.45) (0.5,0.45) -- (1.2,1) node[vertex] {};
\draw (0,-0.35) -- (0,-1) node[vertex] {};
\draw (0,-0.35) -- (-0.5,0.45) -- (0.5,0.45) -- cycle;
%
}
=\frac{[6][10]^2[15]}{[2][5][30]}
\left(
\frac{[32]}{[3][16]}+\frac{[36]}{[9][12]}
\right)
\U
\end{gather}
(the mirror image of the last identity follows from it).

At this stage we have all we need to produce a basis of intertwiners
of endomorphisms of $V^{\otimes 2}$; we propose
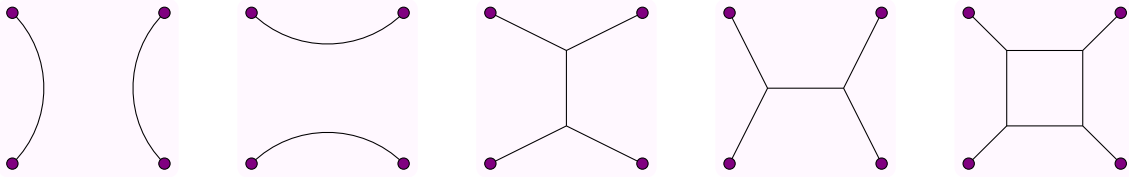
\[
\hspace{-6mm}
\full
\qquad
\cupcap
\qquad
\vertical
\qquad
\horizontal
\qquad
\mysquare
\]
It is a consequence of what follows that these intertwiners are indeed linearly
independent.

Now let us work out the multiplication table in this basis. Note that the last
basis element is simply the square of the fourth.
Thanks to \eqref{eq:rela}--\eqref{eq:relc}, we can compute all products except
the product of the last two (which itself allows to compute the square of the last).
This identity is not particularly nice, and is given in Appendix~\ref{app:mult} for
the sake of completeness. We present in the next section a more pleasant
alternative.

Before we proceed, let us point out that the trivial representation, being
the same as $V^{\otimes 0}$, is diagrammatically\dots invisible. However, in what follows
it is sometimes useful to emphasize that in a tensor product some trivial factors
are included, in which case we draw an empty vertex; e.g.,
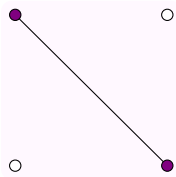
\[
\tikz{
\draw (1,-1) node[vertex] {} -- (-1,1) node[vertex] {};
\node[empty] at (1,1) {};
\node[empty] at (-1,-1) {};
}
\]
is the natural map from $V\otimes \CC$ to $\CC\otimes V$.

\subsection{The spectral parameter-independent $R$-matrix}\label{sec:R0}
As mentioned before, our diagrams are required to be planar. Let us informally discuss this point.
In the undeformed case of $\e$, it is natural \cite{Cvita-book} to consider 
arbitrary diagrams (i.e., not embedded into the plane); however,  one can always 
turn them into planar diagrams by projecting them onto the plane, thus creating new ``virtual crossings''
where two formerly nonintersecting lines are now on top of each other. 
These crossings simply correspond to the natural
permutation of tensors $x\otimes y\mapsto y\otimes x$, which commutes with the $\e$ action. The point is of
course that such a permutation does not commute with the $\Uq(\e)$ action; instead one must use
the $R$-matrix (of $\Uq(\e)$).

With a bit of foresight, we therefore define
\begin{multline}\label{eq:defR0}
\crossing=f(q)\ 
\full
+
f(q^{-1})
\cupcap
\\
+g(q)
\vertical
+g(q^{-1})
\horizontal
+h(q)
\mysquare
\end{multline}
where
\begin{align*}
f(q)&=\frac{[6][10][15]^2}{[3][5][30]^2}\left({q^{-24}}-{q^{-22}}-{q^{-18}}+{q^{-16}}-{q^{-14}}+q^{-8}-1+q^6-q^{12}\right)
\\
g(q)&=\frac{[15]}{[3][30]}\left(
-q^{-16}-q^{-14}-q^{-12}-q^{-8}+q^{-6}+q^{-2}-1+q^2+q^6-q^8-q^{14}-q^{16}-q^{18}
\right)
\\
h(q)&=h(q^{-1})=\frac{[5]}{[6][10]}
\end{align*}
Rather than deriving this ($\Uq(\e)$, i.e., spectral parameter-independent) $R$-matrix now, we postpone
its justification to \S\ref{sec:pties}, where the full ($\Uq(\ee)$, i.e., spectral parameter-dependent)
$R$-matrix will be obtained.
One can also define a similar diagram where overcrossing and undercrossing
are switched, by a simple 90 degree rotation; equivalently, it corresponds
to switching $q\leftrightarrow q^{-1}$ in the coefficients of the right hand side.

These two diagrams are inverses of each other so that one has the relations
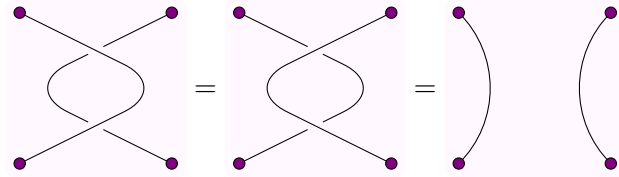
\[
\tikz[rounded corners=7mm,baseline=-3pt]{
\draw (-1,-1) node[vertex] {} -- (1,0) -- (-1,1) node[vertex] {};
\draw (1,-1) node[vertex] {} -- (0.1,-0.55) (-0.1,-0.45) -- (-1,0) -- (-0.1,0.45) (0.1,0.55) -- (1,1) node[vertex] {};
}
=
\tikz[rounded corners=7mm,baseline=-3pt]{
\draw (1,-1) node[vertex] {} -- (-1,0) -- (1,1) node[vertex] {};
\draw (-1,-1) node[vertex] {} -- (-0.1,-0.55) (0.1,-0.45) -- (1,0) -- (0.1,0.45) (-0.1,0.55) -- (-1,1) node[vertex] {};
}
=
\full
\]
which one can recognize as Reidemeister move II. We also mention for the sake of completeness
the Reidemeister move III (or braid relation)
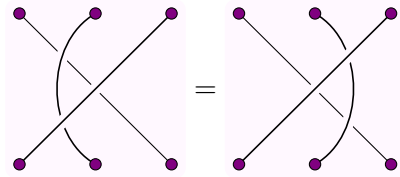
\[
\tikz[baseline=-3pt]{
\draw (1,-1) node[vertex] {} -- (-1,1) node[vertex] {};
\draw[bend left=60,draw=bgcolor,double=black,line width=2pt] (0,-1) node[vertex] {} to (0,1) node[vertex] {};
\draw[draw=bgcolor,double=black,line width=2pt] (-1,-1) node[vertex] {} -- (1,1) node[vertex] {};
}
=
\tikz[baseline=-3pt]{
\draw (1,-1) node[vertex] {} -- (-1,1) node[vertex] {};
\draw[bend right=60,draw=bgcolor,double=black,line width=2pt] (0,-1) node[vertex] {} to (0,1) node[vertex] {};
\draw[draw=bgcolor,double=black,line width=2pt] (-1,-1) node[vertex] {} -- (1,1) node[vertex] {};
}
\]

\rem[gray]{is there a useful knot theory application?}

We can now use the alternative basis of
$\text{End}(V^{\otimes 2})^{\Uq(\e)}$:
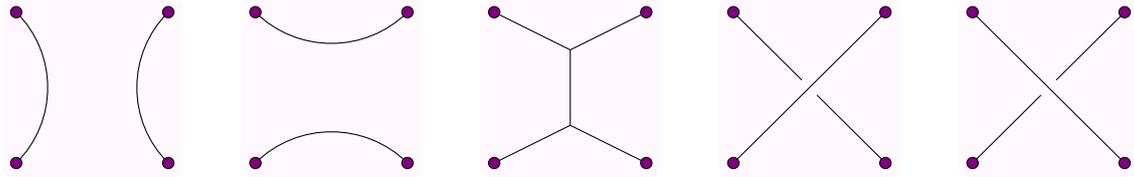
\[
\full
\qquad
\cupcap
\qquad
\vertical
\qquad
\crossing
\qquad
\invcrossing
\]

In order to complete the multiplication table, we need the following additional relations:
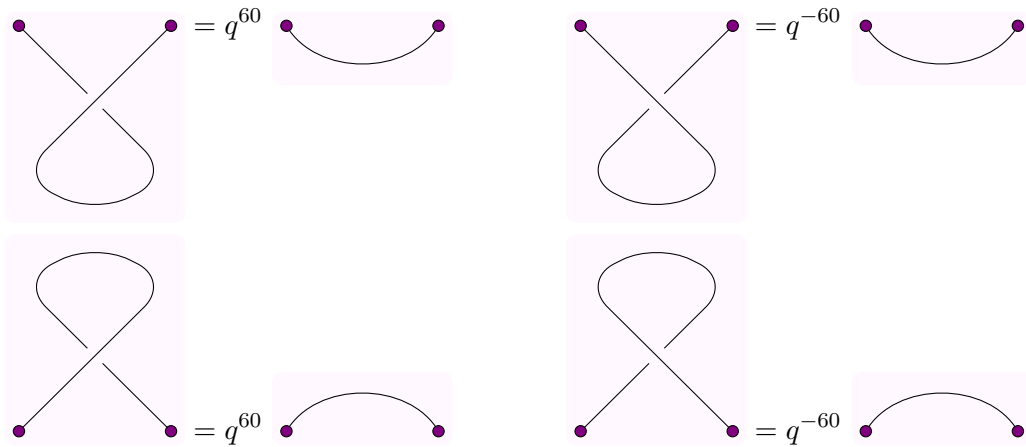
\begin{align*}
\tikz[baseline=-3pt,rounded corners=5mm]{
\draw (-1,0) node[vertex] {} -- (-0.1,-0.9) (0.1,-1.1) -- (1,-2) -- (0,-2.5) -- (-1,-2) -- (1,0) node[vertex] {};
}
&=
q^{60}\ \mycup
&
\tikz[baseline=-3pt,rounded corners=5mm]{
\draw (-1,0) node[vertex] {} -- (1,-2) -- (0,-2.5) -- (-1,-2) -- (-0.1,-1.1) (0.1,-0.9) -- (1,0) node[vertex] {};
}
&=
q^{-60}\ \mycup
\\
\tikz[baseline=-3pt,rounded corners=5mm]{
\draw (-1,0) node[vertex] {} -- (1,2) -- (0,2.5) -- (-1,2) -- (-0.1,1.1) (0.1,0.9) -- (1,0) node[vertex] {};
}
&=
q^{60}\ \mycap
&
\tikz[baseline=-3pt,rounded corners=5mm]{
\draw (-1,0) node[vertex] {} -- (-0.1,0.9) (0.1,1.1) -- (1,2) -- (0,2.5) -- (-1,2) -- (1,0) node[vertex] {};
}
&=
q^{-60}\ \mycap
\end{align*}
which is nothing but Reidemeister move I, as well as 
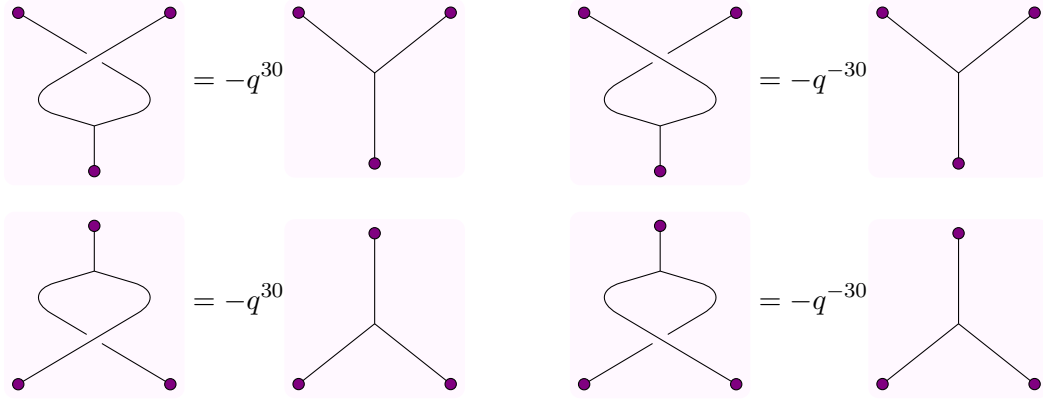
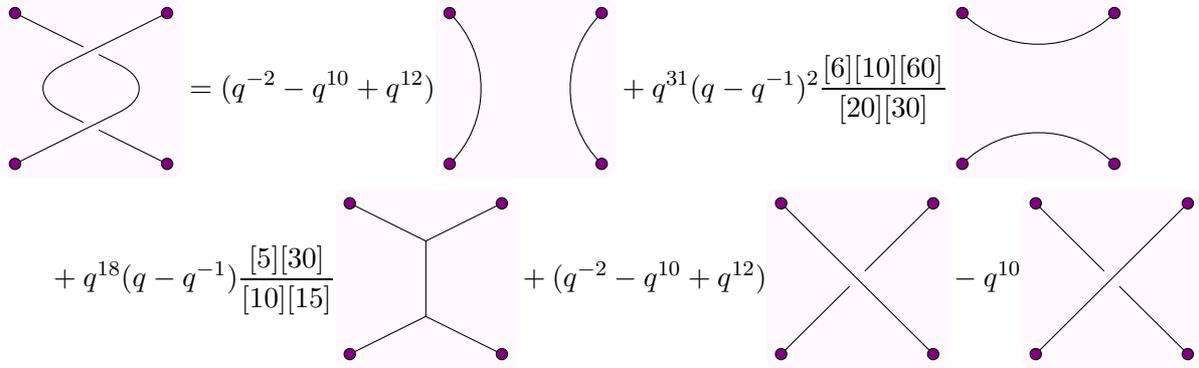
\begin{align*}
\tikz[baseline=-0.4cm,rounded corners=5mm,yscale=0.6]{
\draw (-1,1) node[vertex] {} -- (-0.1,0.1) (0.1,-0.1) -- (1,-1) {[sharp corners] -- (0,-1.5)} -- (-1,-1) -- (1,1) node[vertex] {};
\draw (0,-1.5) -- (0,-2.5) node[vertex] {};
}
&=
-q^{30}
\U
&
\tikz[baseline=-0.4cm,rounded corners=5mm,yscale=0.6]{
\draw (-1,1) node[vertex] {} -- (1,-1) {[sharp corners] -- (0,-1.5)} -- (-1,-1) -- (-0.1,-0.1) (0.1,0.1) -- (1,1) node[vertex] {};
\draw (0,-1.5) -- (0,-2.5) node[vertex] {};
}
&=
-q^{-30}
\U
\\[2mm]
\tikz[baseline=0.4cm,rounded corners=5mm,yscale=0.6]{
\draw (-1,-1) node[vertex] {} -- (1,1) {[sharp corners] -- (0,1.5)} -- (-1,1) -- (-0.1,0.1) (0.1,-0.1) -- (1,-1) node[vertex] {};
\draw (0,1.5) -- (0,2.5) node[vertex] {};
}
&=
-q^{30}
\D
&
\tikz[baseline=0.4cm,rounded corners=5mm,yscale=0.6]{
\draw (-1,-1) node[vertex] {} -- (-0.1,-0.1) (0.1,0.1) -- (1,1) {[sharp corners] -- (0,1.5)} -- (-1,1) -- (1,-1) node[vertex] {};
\draw (0,1.5) -- (0,2.5) node[vertex] {};
}
&=
-q^{-30}
\D
\end{align*}
and
\begin{multline*}
\tikz[rounded corners=7mm,baseline=-3pt]{
\draw (-1,-1) node[vertex] {} -- (1,0) -- (0.1,0.45) (-0.1,0.55) -- (-1,1) node[vertex] {};
\draw (1,-1) node[vertex] {} -- (0.1,-0.55) (-0.1,-0.45) -- (-1,0) -- (1,1) node[vertex] {};
}
=(q^{-2}-q^{10}+q^{12})\full
+q^{31}(q-q^{-1})^2 \frac{[6][10][60]}{[20][30]} \cupcap
\\+
q^{18}(q-q^{-1})\frac{[5][30]}{[10][15]}\vertical
+(q^{-2}-q^{10}+q^{12})\invcrossing
-q^{10}\crossing
\end{multline*}
(and the same identity with undercrossings $\leftrightarrow$ overcrossings,
$q\leftrightarrow q^{-1}$).

\subsection{Primitive idempotents}
Because $V^{\otimes 2}$ is multiplicity-free, 
$\text{End}(V^{\otimes 2})^{\Uq(\e)}$ is a commutative algebra, and multiplication
operators can be diagonalized simultaneously, with eigenvectors the primitive
idempotents (projectors $P_\omega$ onto irreducible subrepresentations $V_\omega$). 
The eigenvalues give the coefficients of the expansion:
\begin{align*}
\full&=P_0+P_{\omega_1}+P_{\omega_7}+P_{2\omega_1}+P_{\omega_2}
\\
\cupcap&=\frac{[20][24][31]}{[6][10]}P_0
\\
\vertical&=\frac{[10][15]^2[18][32]}{[5][9][16][30]}P_{\omega_1}
\\
\crossing&=q^{60}P_0-q^{30}P_{\omega_1}+q^{12}P_{\omega_7}+q^{-2}P_{2\omega_1}-P_{\omega_2}
\end{align*}
and the expansion of the inverse crossing is obtained by inverting all coefficients, i.e.,
$q\leftrightarrow q^{-1}$ as usual.
\rem{are these the correct eigenvalues of the $R$-matrix?}

The multiplication rule for the diagrams is equivalent to the statement
\[
P_{\omega}P_{\omega'}=\delta_{\omega,\omega'} P_{\omega}
\qquad
\omega,\omega'\in\Omega
\]

We also mention the expansion:
\begin{multline*}
\horizontal=\frac{[10][15]^2[18][32]}{[5][9][16][30]}P_0
+\frac{[10]^2[15]}{[30]}\left((q-q^{-1})^2\frac{[7][12]}{[4]}+\frac{[6]^2}{[2]^2[3][5]}\right)
P_{\omega_1}\\
+\frac{[6][10][15][32]}{[5][16][30]}P_{\omega_7}-\frac{[10][15][18]}{[5][9][30]}P_{2\omega_1}+\frac{(q-q^{-1})^2[6][10][15]}{[30]}P_{\omega_2}
\end{multline*}
and the expansion of the ``square'' diagram is obtained by squaring all the coefficients above.


\section{\texorpdfstring{$\Uq(\ee)$}{Uq(e8(1))} and its $R$-matrix}\label{sec:aff}
\subsection{The $(248+1)$-dimensional representation}
In this section we present explicit formulae for the one-parameter
family of representations $W_z$ of dimension $249$ of $\Uq(\ee)$. Note that the parameter $z\in\CC^\times$,
called ``spectral parameter'' comes automatically from the $\ZZ$-grading of
the algebra.

The action of $\Uq(\e)\subset \Uq(\ee)$ on $W_z$ is as follows: 
it is the direct sum of the $248$-dimensional
representation $V$ defined in \S\ref{sec:248} and of the trivial representation. We use the same
basis $\{v_\beta,\ \beta\in\Phi\}\cup\{u_i,\ i=1,\ldots,8\}$ for the former and fix some
nonzero vector $w$ in the latter.

The action of the remaining generators of $\Uq(\ee)$ on $W_1$ is given by
\begin{align*}
k_0 v_\beta &= q^{\braket{\alpha_0}{\beta}} v_\beta
&
k_0 u_j &= u_j
&
k_0 w &= w
\\
e_0 v_\beta &= \begin{cases}
\epsilon(\alpha_0,\beta) v_{\beta+\alpha_0}&\beta+\alpha_0\in\Phi
\\
-c_0\kappa^{-1}w-\sum_{i=1}^8 c_i u_i&\beta+\alpha_i=0
\\
0&\text{else}
\end{cases}
&
e_0 u_j &= -[C_{0j}] v_{\alpha_0} 
&
e_0 w &=\kappa\, v_{\alpha_0}
\\
f_0 v_\beta &= \begin{cases}
\epsilon(-\alpha_0,\beta) v_{\beta-\alpha_0}&\beta-\alpha_0\in\Phi
\\
c_0\kappa^{-1}w+\sum_{i=1}^8 c_i u_i&\beta-\alpha_i=0
\\
0&\text{else}
\end{cases}
&
f_0 u_j &= [C_{0j}] v_{-\alpha_0} 
&
f_0 w &=-\kappa\,v_{-\alpha_0}
\end{align*}
where the $c_i$ are constants to be determined, and $\kappa$ is a parameter related to the freedom
in the relative normalization of $w$ and of the basis of $V$, which we do not fix for now.

Imposing the relations of the algebra is
equivalent to the following system of equations:
\begin{align*}
\sum_{j=1}^8 [C_{ij}] c_j  &= \delta_{i,1} \qquad i=1,\ldots,8
\\
c_0+c_1&=[2]
\end{align*}
which can be readily solved:
\begin{align*}
c_i&=\frac{1}{\det [C]} \widetilde{[C]}_{i1}
=\frac{[6][10][15]}{[2][3][5][30]}
\left(
\frac{[2][3][18]}{[6][9]},
\frac{[2][3][12]}{[4][6]},
\frac{[2][8]}{[4]},
[5],
[2][3],
[2]^2,
[2],
[3]
\right)
\qquad
i=1,\ldots,8
\\
c_0&=(q-q^{-1})^2\frac{[6][10][15]}{[30]}
\end{align*}

We then restore the $z$-dependence to obtain the general action on $W_z$ by
transforming generators according to their degree: $e_0\mapsto z\,e_0$, $f_0\mapsto z^{-1}f_0$,
all others unchanged.

\subsection{The spectral parameter-dependent $R$-matrix}
Because $\Uq(\ee)$ is pseudotriangular, we know
that for generic $z,z'\in \CC^\times$, $W_z\otimes W_{z'}$ and $W_{z'}\otimes W_z$ are isomorphic.
Up to normalization, the (unique) intertwiner is given by the universal $R$-matrix in these representations;
since the universal $R$-matrix is of degree $0$, this intertwiner only depends on $z'/z$.

Here we follow Jimbo's approach \cite{Jimbo-qg-YBE} which is to directly compute the intertwiner from the
intertwining relations. The normalization that we obtain may be different from the one coming
from the universal $R$-matrix.

We thus look for an operator $\cR(z'/z): W_z\otimes W_{z'}\to W_{z'}\otimes W_z$ that satisfies
\begin{equation}\label{eq:inter}
\cR(z'/z)(\rho_z\otimes \rho_{z'})(\Delta(x))=(\rho_{z'}\otimes \rho_z)(\Delta(x))\cR(z'/z)
\end{equation}
for all $x\in \Uq(\ee)$, where in this equation only, we write explicitly the representation map
$\rho_z: \Uq(\ee)\to \text{End}(W_z)$.

First we consider the commutation with $x\in \Uq(\e)$ only. We know that $W_z\cong V_{\omega_1}\oplus V_0$
as a $\Uq(\e)$ representation; using \eqref{eq:dec}, we obtain\footnote{Note in particular
that the ``multiplicity-free'' property that is crucial in e.g.\ \cite{DGZ-YBE} fails here.}
\[
W_z\otimes W_{z'}\cong \CC^2\otimes V_0 \oplus \CC^3\otimes V_{\omega_1}\oplus V_{\omega_7}\oplus V_{2\omega_1}
\oplus V_{\omega_2}
\]
This means that the space of $\Uq(\e)$ intertwiners from $W_z\otimes W_{z'}$ to $W_{z'}\otimes W_z$
is of dimension $2^2+3^2+1^2+1^2+1^2=16$.
Using the diagrammatic language of \S\ref{sec:diag}, it is easy to produce
a basis of this space, namely
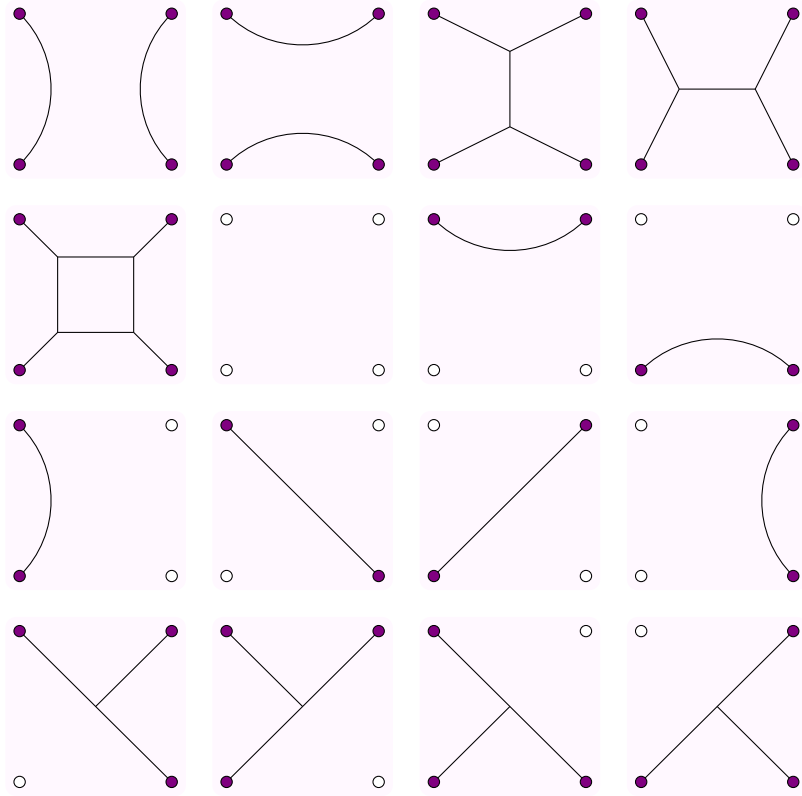
\[
\begin{array}{cccc}
\full
&
\cupcap
&
\vertical
&
\horizontal
\\[13mm]
\mysquare
&
\myempty
&
\cupa
&
\capa
\\[13mm]
\leftid
&
\diagida
&
\diagidb
&
\rightid
\\[13mm]
\Yaa
&
\Yba
&
\Yab
&
\Ybb
\end{array}
\]
Recall that filled (resp.\ empty) vertices correspond to $V_{\omega_1}=V$ (resp.\ $V_0=\CC$), so that
when one concatenates diagrams, if the type of the vertices does not match, the resulting product is
zero.

$\cR(z)$ must therefore be a linear combination of these 16 intertwiners. Imposing \eqref{eq:inter}
for the remaining generators $e_0,f_0,k_0$ of $\Uq(\ee)$ results in a very large number of linear equations
on the entries of $\cR(z'/z)$. For practical purposes, since we know the system has a nontrivial solution,
it is sufficient to extract a subset of $15$ independent equations among these and then solve them. 
Such a subset can be found for any $z,z'$ (as opposed to, for generic $z,z'$), implying uniqueness of the
intertwiner. Finding and solving these equations 
is best performed by computer, and we present the explicit solution in Appendix~\ref{app:main},
where the chosen normalization is
\begin{equation}\label{eq:norm}
\cR(z) v_{\omega_1}\otimes v_{\omega_1}=v_{\omega_1}\otimes v_{\omega_1}
\end{equation}
(i.e., the highest-weight-to-highest-weight entry is $1$).

We now discuss properties of $\cR(z)$. 
These properties should be checked using the diagrammatic algebra itself (which is only of dimension $16$,
so $16\times16$ matrices in the regular representation), rather than using the actual
operators acting on copies of $W_z$. It is sometimes useful to represent graphically the $R$-matrix as
\[
\check R(z'/z)=
\tikz[baseline=-3pt]{
\draw[dashed] (-1,-1) node[halfvertex] {}  -- (1,1) node[halfvertex] {} node[right] {$z'$};
\draw[dashed] (1,-1) node[halfvertex] {} -- node[Rmat] {} (-1,1) node[halfvertex] {} node[left] {$z$};
}
\]
where half-filled vertices correspond to the whole of $W_z\cong V_{\omega_1}\oplus V_0$,
and correspondingly, dashed lines to identity operators on $W_z$,
namely
$\tikz[baseline=-3pt]{\draw[dashed] (0,1) node[halfvertex] {} -- (0,-1) node[halfvertex] {};}
=\tikz[baseline=-3pt]{\draw (0,1) node[vertex] {} -- (0,-1) node[vertex] {};}
+\tikz[baseline=-3pt]{\path (0,1) node[empty] {} -- (0,-1) node[empty] {};}$.

\subsection{Basic properties of the $R$-matrix}\label{sec:pties}
First, from Schur's lemma one knows that $\cR(z) \cR(z^{-1})$ must be proportional to the identity.
With the chosen normalization \eqref{eq:norm}, one in fact has the so-called unitarity equation
\begin{align}\label{eq:unit}
\tikz[rounded corners=7mm,baseline=-3pt]{
\draw[dashed] (-1,-1) node[halfvertex] {} -- node[Rmat] {} (1,0) -- node[Rmat] {} (-1,1) node[halfvertex] {};
\draw[dashed] (1,-1) node[halfvertex] {} -- (-1,0) -- (1,1) node[halfvertex] {};
}
&=
\full+\leftid+\rightid+\myempty
\\\notag&=\identity
\end{align}
As a check of the methods of this paper, the reader is invited
to verify this equation by direct computation.

Similarly, one can check the Yang--Baxter equation
\begin{align}\label{eq:YBE}
(\cR(z)\otimes 1)
(1\otimes\cR(wz))
(\cR(w)\otimes 1)
&=
(1\otimes\cR(w))
(\cR(wz)\otimes 1)
(1\otimes\cR(z))
\\\notag
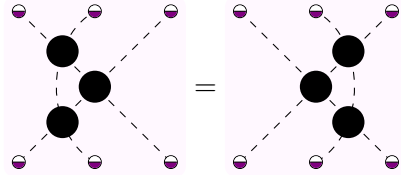
\tikz[baseline=-3pt,dashed]{
\draw (1,-1) node[halfvertex] {} -- node[Rmat] {} (-1,1) node[halfvertex] {};
\draw[bend left=60] (0,-1) node[halfvertex] {} to node[Rmat,pos=0.3] {} node[Rmat,pos=0.7] {} (0,1) node[halfvertex] {};
\draw (-1,-1) node[halfvertex] {} -- (1,1) node[halfvertex] {};
}
&=
\tikz[baseline=-3pt,dashed]{
\draw (1,-1) node[halfvertex] {} -- node[Rmat] {} (-1,1) node[halfvertex] {};
\draw[bend right=60] (0,-1) node[halfvertex] {} to node[Rmat,pos=0.3] {} node[Rmat,pos=0.7] {} (0,1) node[halfvertex] {};
\draw (-1,-1) node[halfvertex] {} -- (1,1) node[halfvertex] {};
}
\end{align}

Next we come to the so-called crossing symmetry. Consider the 90 degree counterclockwise rotation of
the $R$-matrix. Then we expect it to be related to the original $R$-matrix up to shift of the
spectral parameter. The 90 degree rotation simply permutes our basis of $\Uq(\e)$ intertwiners,
so it easy to compute the result. More specifically, because the dual representation of $W_z$
satisfies $W_z^*\cong W_{q^{30}z}$, we expect the two diagrams below to be proportional:
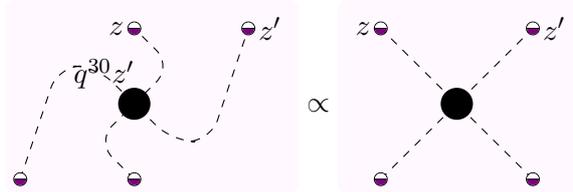
\[
\tikz[baseline=-3pt,rounded corners=3mm]{
\draw[dashed] (-1.5,-1) node[halfvertex] {}  -- (-1,0.5) -- (-0.5,0.5) -- node[pos=0.1] {$q^{30}z'$} (0.5,-0.5) -- (1,-0.5) -- (1.5,1) node[halfvertex] {} node[right] {$z'$};
\draw[dashed] (0,-1) node[halfvertex] {} -- (-0.5,-0.5) -- node[Rmat] {} (0.5,0.5) -- (0,1) node[halfvertex] {} node[left] {$z$};
}
\propto
\tikz[baseline=-3pt]{
\draw[dashed] (-1,-1) node[halfvertex] {}  -- (1,1) node[halfvertex] {} node[right] {$z'$};
\draw[dashed] (1,-1) node[halfvertex] {} -- node[Rmat] {} (-1,1) node[halfvertex] {} node[left] {$z$};
}
\]
For this to happen, it is clear from the expression of the $R$-matrix
in appendix~\ref{app:main} that one must define dashed cups and caps as follows:
\begin{align}\label{eq:dashcap}
\tikz[baseline=-3pt]{\draw[bend left=60,dashed] (-1,0) node[vertex] {} to (1,0) node[vertex] {};}
&=
\mycap+\frac{(q-q^{-1})^2[6][10][15]}{[30]\kappa^2}
\ \tikz[baseline=-3pt]{
\node[empty] at (-1,0) {};
\node[empty] at (1,0) {};
}
\\\label{eq:dashcup}
\tikz[baseline=-3pt]{\draw[bend left=60,dashed] (1,0) node[vertex] {} to (-1,0) node[vertex] {};}
&=
\mycup+\frac{[30]\kappa^2}{(q-q^{-1})^2[6][10][15]}
\ \tikz[baseline=-3pt]{
\node[empty] at (-1,0) {};
\node[empty] at (1,0) {};
}
\end{align}
Of course, it is very natural at this stage to set the free parameter $\kappa$ to the value
\begin{equation}\label{eq:kappa}
\kappa=(q-q^{-1})\left(\frac{[6][10][15]}{[30]}\right)^{1/2}
\end{equation}
\rem[gray]{i.e., $\kappa=\sqrt{c_0}$, which kinda makes sense}
and we do so in the remainder of this paper. The exact crossing relation is
\begin{equation}\label{eq:crossing}
\tikz[baseline=-3pt,rounded corners=3mm]{
\draw[dashed] (-1.5,-1) node[halfvertex] {}  -- (-1,0.5) -- (-0.5,0.5) -- node[pos=0.1] {$q^{30}z'$} (0.5,-0.5) -- (1,-0.5) -- (1.5,1) node[halfvertex] {} node[right] {$z'$};
\draw[dashed] (0,-1) node[halfvertex] {} -- (-0.5,-0.5) -- node[Rmat] {} (0.5,0.5) -- (0,1) node[halfvertex] {} node[left] {$z$};
}
=q^4\frac{(1-z)(1-q^{10}z)(1-q^{18}z)(1-q^{28}z)}{(1-q^2z)(1-q^{12}z)(1-q^{20}z)(1-q^{30}z)}
\tikz[baseline=-3pt]{
\draw[dashed] (-1,-1) node[halfvertex] {}  -- (1,1) node[halfvertex] {} node[right] {$z'$};
\draw[dashed] (1,-1) node[halfvertex] {} -- node[Rmat] {} (-1,1) node[halfvertex] {} node[left] {$z$};
}
\end{equation}
Note that the prefactor can be absorbed in the redefinition 
\[
\check R_{\mathrm{poly}}(z)=z^{-2}(1-q^2z)(1-q^{12}z)(1-q^{20}z)(1-q^{30}z)\check R(z)
\]
(which turns entries into Laurent polynomials of $z$), 
but then it is the unitarity equation \eqref{eq:unit} which acquires a prefactor.\footnote{%
Requiring both unitarity and crossing relations without a prefactor would take us beyond the realm
of rational functions, and we do not pursue this here.}

There are three special values of the spectral parameter. First, at $z=1$, 
one has
$
\cR(1)=1
$.
The other two values are $0$ and $\infty$:
\begin{align*}
\cR(0)&=q^2\left(\crossing+\diagida+\diagidb+\myempty\right)
\\
\cR(\infty)&=q^{-2}\left(\invcrossing+\diagida+\diagidb+\myempty\right)
\end{align*}
\rem{why these annoying factors of $q^{\pm2}$?}
where the crossing is defined in \eqref{eq:defR0}.
This gives an {\em a posteriori}\/ justification for the introduction of the crossing in
\S\ref{sec:R0}.

At this stage it is natural to switch to the alternative basis which involves these crossings,
and the expression of the $R$-matrix simplifies. We find, using \eqref{eq:kappa} to further simplify,
\begin{align*}
\cR(z)&=
q^3(q-q^{-1})^2[6]\frac{z}{(1-q^2z)(1-q^{12}z)}
\full[0.85]
\\
&\ 
+q^3\frac{\kappa^2}{[5]}\frac{z(1-z)(1-q^{10}-q^{20}z+q^{30}+q^{40}z-q^{50}z)}{(1-q^2z)(1-q^{12}z)(1-q^{20}z)(1-q^{30}z)}
\cupcap[0.85]
\\
&\ 
-q^{13}(q-q^{-1})^2[5]\frac{z(1-z)}{(1-q^2z)(1-q^{12}z)(1-q^{20}z)}
\vertical[0.85]
\\
&\ 
+q^7\frac{1-z}{(1-q^2z)(1-q^{12}z)}
\left(
q^{-5}\crossing[0.85]
-q^{5}z
\invcrossing[0.85]
\right)
\\
&\ 
+q^2(q-q^{-1})\kappa^2 \frac{z(1+q^{30}z)}{(1-q^2z)(1-q^{12}z)(1-q^{20}z)}
\left(
\leftid[0.85]+\rightid[0.85]
\right)
\\
&\ 
+\left(\cf{\leftid[0.28]}+\frac{q^2-z}{1-q^2z}\right)\left(\diagida[0.85]+\diagidb[0.85]\right)
\\
&\hspace{-2cm}
+q^{17}(q-q^{-1})\kappa\frac{z(1-z)}{(1-q^2z)(1-q^{12}z)(1-q^{20}z)}
\left(
\Yaa[0.85]+\Yba[0.85]
+\Yab[0.85]+\Ybb[0.85]
\right)
\\
&\hspace{-2cm}
+q^{32}(q-q^{-1})\kappa^2\frac{z(1-z^2)}{(1-q^2z)(1-q^{12}z)(1-q^{20}z)(1-q^{30}z)}
\left(\capa[0.85]
+\cupa[0.85]
\right)
\\
&\hspace{-2cm}
+\left(\cf{\full[0.28]}+\cf{\cupcap[0.28]}
+\cf{\crossing[0.28]}+\cf{\invcrossing[0.28]}
-(q-q^{-1})\frac{[6][15]}{[5]\kappa}\cf{\Yaa[0.28]}
\right)
\myempty[0.85]
\end{align*}
where we have reused certain coefficients by denoting them $\cf{\cdot}$.

\subsection{Fusion channels}
In the three multiplicity-free channels, one can directly rewrite the $R$-matrix in terms
of the idempotents:
\begin{multline}
\cR(z)=P_{2\omega_1}
+\frac{z-q^2}{1-q^2z}P_{\omega_2}
+\frac{(z-q^2)(z-q^{12})}{(1-q^2 z) (1-q^{12} z)}P_{\omega_7}
\\
+\frac{1}{(1-q^2 z) (1-q^{12} z)(1-q^{20})} \cR_{\omega_1}(z)
+\frac{1}{(1-q^2 z) (1-q^{12} z)(1-q^{20})(1-q^{30})}\cR_0(z)
\end{multline}
where the last two terms are the restrictions of $\cR(z)$
to $\CC^3\otimes V_{\omega_1}$ and $\CC^2\otimes V_0$, respectively, in which we have taken out the LCM of
denominators (and the GCD of numerators is $1$). Diagonalizing $\check R(z)$ in these remaining
blocks results in an eigenvalue 
$\frac{z-q^2}{1-q^2z}$ in the $V_{\omega_1}$ block, and then pairs of eigenvalues that are
solutions of quadratic equations in each block.
\rem{why the extra degeneracy between $V_{\omega_2}$ and $V_{\omega_1}$? some non QG invariant extra symmetry?}

As a function of the spectral parameter $z$, $\cR(z)$ is a rational function with
poles at $z=q^{-2},q^{-12},q^{-20},q^{-30}$. Conversely, due to unitarity equation~\eqref{eq:unit},
at $z=q^2,q^{12},q^{20},q^{30}$, $\cR(z)$ is noninvertible. More precisely, define
\[
\begin{aligned}
f_a&=\mathrm{Res}_{z=q^{-a}} \cR(z)
\\
g_a&=\cR(z=q^a)
\end{aligned}
\qquad
a\in \{2,12,20,30\}
\]

Then from \eqref{eq:unit} we have $f_ag_a=g_af_a=0$, and in fact, expanding this equation around
$z=q^a$, diagonalizing $\cR(z)$ and noting that each zero eigenvalue of $\cR(z)$ at $z=q^a$ corresponds to a 
pole eigenvalue of $\cR(z)$ at $z=q^{-a}$, we conclude that their ranks are complementary,
so that $\textrm{Im}(f_a)=\textrm{Ker}(g_a)=:X_a$ and $\textrm{Im}(g_a)=\textrm{Ker}(f_a)=:Y_a$.
$X_a$ is a submodule of $W_{z}\otimes W_{q^a z}$ with quotient isomorphic to $Y_a$, and vice versa
for $W_{q^a z}\otimes W_{z}$. Because $\cR(z'/z)$ is the {\em only}\/ intertwiner from 
$W_z\otimes W_{z'}$ to $W_{z'}\otimes W_z$ even at $z'/z=q^{\pm a}$, 
$W_z\otimes W_{q^{\pm a}z}$ is indecomposable. \rem[gray]{can we say more? irreducibility of $X$?}

If $a=30$, $X_{30}$ is of dimension $1$: we recognize it as the trivial representation of $\Uq(\ee)$.
As an exercise, let us recover this by looking at possible embeddings of the trivial representation into
$W_z\otimes W_{z'}$. From the $\Uq(\e)$ invariance, we know that this embedding must take the form
\begin{equation}\label{eq:triv1}
\alpha\ \mycap + \beta\ 
\tikz[baseline=-3pt]{
\node[empty] at (-1,0) {};
\node[empty] at (1,0) {};
}
\end{equation}

Imposing the invariance under $e_0$ results in the relations
\[
z'=q^{30}z\qquad 
\alpha=\beta
\]
which is consistent with the above. Of course, we recognize that \eqref{eq:triv1} is (up to normalization)
the dashed cap \eqref{eq:dashcap},
and its invariance is once again nothing but the statement that $W_z^*\cong W_{q^{30}z}$.

Dually, the invariance of the  $W_z\otimes W_{z'}\to \CC$ diagram
\begin{equation}\label{eq:triv2}
\alpha\ \mycup + \beta\ \tikz[baseline=-3pt]{
\node[empty] at (-1,0) {};
\node[empty] at (1,0) {};
}
\end{equation}
leads to
\[
z'=q^{-30}z\qquad 
\alpha=\beta
\]
(compare with the dashed cup \eqref{eq:dashcup}).
Up to normalization, $f_a$ is just the product of \eqref{eq:triv1} and \eqref{eq:triv2}.

If $a=20$, $X_{20}$ is of dimension $249$, and so it must be isomorphic to
$W_{z''}$ for some $z''$. Again, we look for an embedding of $W_{z''}$ into $W_{z}\otimes W_{z'}$; 
by $\Uq(\e)$ invariance it is of the form
\begin{equation}\label{eq:fus1}
\alpha\ \D+\beta\ \tikz[baseline=-3pt]{
\draw[bend left=60] (-1,-1) node[vertex] {} to (1,-1) node[vertex] {};
\node[empty] at (0,1) {};
}
+\gamma\ 
\tikz[baseline=-3pt]{
\draw[bend right=40] (-1,-1) node[vertex] {} to (0,1) node[vertex] {};
\node[empty] at (1,-1) {};
}
+\delta\ 
\tikz[baseline=-3pt]{
\draw[bend left=40] (1,-1) node[vertex] {} to (0,1) node[vertex] {};
\node[empty] at (-1,-1) {};
}
+\epsilon\ 
\tikz[baseline=-3pt]{
\node[empty] at (1,-1) {};
\node[empty] at (-1,-1) {};
\node[empty] at (0,1) {};
}
\end{equation}
and requiring that it commute with the $e_0$ action leads to
\[
z'=q^{20}z\qquad z''=q^{10}z\qquad
\frac{\alpha}{\epsilon}=
-\frac{[5]}{[20]}\sqrt{\frac{[10][30]}{[6][15]}}
\qquad
\frac{\beta}{\epsilon}=
\frac{\gamma}{\epsilon}=\frac{\delta}{\epsilon}=\frac{[10]}{[20]}
\]
so that $W_z\otimes W_{q^{20}z}\supset X_{20}\cong W_{q^{10}z}$.

Dually, one finds
\begin{equation}\label{eq:fus2}
\alpha\ \U+\beta\ \tikz[baseline=0cm,yscale=-1]{
\draw[bend left=60] (-1,-1) node[vertex] {} to (1,-1) node[vertex] {};
\node[empty] at (0,1) {};
}
+\gamma\ 
\tikz[baseline=0cm,yscale=-1]{
\draw[bend right=40] (-1,-1) node[vertex] {} to (0,1) node[vertex] {};
\node[empty] at (1,-1) {};
}
+\delta\ 
\tikz[baseline=0cm,yscale=-1]{
\draw[bend left=40] (1,-1) node[vertex] {} to (0,1) node[vertex] {};
\node[empty] at (-1,-1) {};
}
+\epsilon\ 
\tikz[baseline=0cm,yscale=-1]{
\node[empty] at (1,-1) {};
\node[empty] at (-1,-1) {};
\node[empty] at (0,1) {};
}
\end{equation}
and requiring that it commute with the $e_0$ action leads to
\[
z'=q^{-20}z\qquad z''=q^{-10}z
\]
and the same ratios of parameters.
Again, $f_a$ is proportional to the product of \eqref{eq:fus1} and \eqref{eq:fus2}.

For $a=2,12$, the spaces $X_a$ are higher-dimensional, and are not so simple to describe in our diagrammatic
language. In terms of $\Uq(\e)$, one has:
\begin{align*}
X_{12}&=V_{\omega_7}\oplus V_{\omega_1}\oplus V_0
\\
X_2&=V_{\omega_2}\oplus V_{\omega_7}\oplus \CC^2\otimes V_{\omega_1}\oplus V_0
\end{align*}
\rem[gray]{in the first case, it has to be irreducible, otherwise sub/quotient would be one of the previous
two cases}
By obvious complementation, one can get similar decompositions for the spaces $Y_a$, but again, their description
is more complicated. For instance, we have
\[
Y_2\cong V_{2\omega_1}\oplus V_{\omega_1}\oplus V_0
\]
and we can describe its embedding into $W_{q^2z}\otimes W_z$ as
\begin{multline*}
\alpha\dbl[1]
+ \beta
\left(
\D[1]
+
\frac{[32]}{[16]}\sqrt{\frac{[6][10][15]}{[30]}}
\left(
\tikz[baseline=-3pt,scale=1]{
\draw[bend right=40] (-1,-1) node[vertex] {} to (0,1) node[vertex] {};
\node[empty] at (1,-1) {};
}
+
\tikz[baseline=-3pt,scale=1]{
\draw[bend left=40] (1,-1) node[vertex] {} to (0,1) node[vertex] {};
\node[empty] at (-1,-1) {};
}
\right)
\right)
\\
+\gamma 
\left(
\tikz[baseline=-3pt,scale=1]{
\draw[bend left=60] (-1,-1) node[vertex] {} to (1,-1) node[vertex] {};
\node[empty] at (0,1) {};
}
+ [31]\ 
\tikz[baseline=-3pt,scale=1]{
\node[empty] at (1,-1) {};
\node[empty] at (-1,-1) {};
\node[empty] at (0,1) {};
}
\right)
\end{multline*}
where $\dbl[0.5]$ is the embedding of $V_{2\omega_1}$ into $V_{\omega_1}^{\otimes 2}$ (with some normalization).
However, the coefficients $\alpha,\beta,\gamma$ can only be fixed by independently specifying the
$\Uq(\ee)$ action on $V_{2\omega_1}\oplus V_{\omega_1}\oplus V_0$.

\appendix
\section{Some product rules}\label{app:mult}
\begin{multline*}
\tikz[baseline=-3pt]{
\draw (1,-1) node[vertex] {} -- (0.4,-0.4) [rounded corners=4mm] -- (0.75,0.15)-- (-1,1) node[vertex] {};
\draw[draw=bgcolor,double=black,line width=2pt] (-1,-1) node[vertex] {} -- (-0.4,-0.4) [rounded corners=4mm] -- (-0.75,0.15) -- (1,1) node[vertex] {};
\draw (-0.4,-0.4) -- (0.4,-0.4);
}
=
\tikz[baseline=-3pt,scale=-1]{
\draw (1,-1) node[vertex] {} -- (0.4,-0.4) [rounded corners=4mm] -- (0.75,0.15)-- (-1,1) node[vertex] {};
\draw[draw=bgcolor,double=black,line width=2pt] (-1,-1) node[vertex] {} -- (-0.4,-0.4) [rounded corners=4mm] -- (-0.75,0.15) -- (1,1) node[vertex] {};
\draw (-0.4,-0.4) -- (0.4,-0.4);
}
=
q^{10}\horizontal
-q^{20}\vertical
\\
+\frac{(q-q^{-1})[6][10][15]}{[5][30]}
\left(
q^{-5}\full
+q^{35}\cupcap
+q^{15}\crossing
\right)
\end{multline*}

{\tiny

\def\qn(#1){[#1]}
\begin{align*}
&
\doublesquare=
-\frac{(q^2-1)^2 \qn(6)^2 \qn(10)^3 \qn(15)^3 \qn(18) \qn(32)}{q^2 \qn(5)^2 \qn(9) \qn(16)
   \qn(30)^3}
\full\\
&+\frac{\left(q^{32}-q^{30}+q^{28}-q^{26}+q^{22}-q^{20}+q^{18}-q^{16}+q^{14}-q^{12}+q^{10}-q^6+q^4-q^2+1\right) \qn(2) \qn(6) \qn(10)^3 \qn(12) \qn(15)^3 \qn(18)
   \qn(32)}{q^{16} \qn(4) \qn(5)^2 \qn(9) \qn(16) \qn(30)^3}
\cupcap\\
&+\scriptstyle\frac{\left(q^{20}+q^{18}+q^{16}+q^{12}+q^8+q^4+q^2+1\right)
   \left(q^{24}-q^{18}+q^{16}-q^{14}+q^{12}-q^{10}+q^8-q^6+1\right) \left(q^{36}+q^{30}-q^{26}+q^{24}-q^{20}+q^{18}-q^{16}+q^{12}-q^{10}+q^6+1\right) \qn(6)^2 \qn(10)^2 \qn(15)^2}{q^{40}
   \qn(3) \qn(5)^2 \qn(30)^2}\\
&\hspace{14cm}\vertical\\
&-\frac{\left(q^{48}-q^{46}-q^{42}+q^{40}-q^{38}-q^{36}+q^{32}+q^{30}-q^{28}+q^{26}-2 q^{24}+q^{22}-q^{20}+q^{18}+q^{16}-q^{12}-q^{10}+q^8-q^6-q^2+1\right)
   \qn(6)^2 \qn(10)^2 \qn(15)^2}{q^{24} \qn(3) \qn(5)^2 \qn(30)^2}\\
&\hspace{14cm}\horizontal\\
&+\frac{\left(q^{20}+q^{18}-q^{14}-q^{12}+q^{10}-q^8-q^6+q^2+1\right) \qn(6) \qn(10)^2
   \qn(15)}{q^{10} \qn(2) \qn(3) \qn(5) \qn(30)}
\mysquare
\end{align*}
}
\vfill\eject%
\begin{landscape}%
\section{The main formula}\label{app:main}%
{\def\qn(#1){[#1]}%
\tiny\begin{gather*}
\cR(z)=
 \frac{\qn(6) \qn(10) \qn(15)^2 (q^{58} z^2-q^{56} z^2-q^{52} z^2+q^{50} z^2-q^{48}
   z^2+q^{46} z+q^{42} z^2-q^{40} z+q^{38} z-q^{36}-q^{34} z^2+q^{34} z-q^{32} z
}{}
\\\frac{+q^{30}+q^{28} z^2-q^{26} z+q^{24}
   z-q^{24}-q^{22} z^2+q^{20} z-q^{18} z+q^{16}+q^{12} z-q^{10}+q^8-q^6-q^2+1)}{q^{22} \qn(3)
   \qn(5) \qn(30)^2 (q^2 z-1) (q^{12} z-1)} \full[0.35]\\
 -\frac{\qn(6) \qn(10) \qn(15)^2 (z-1) (q^{10} z-1) (q^{74} z^2-q^{68}
   z^2+q^{62} z^2-q^{54} z^2-q^{54} z+q^{48} z^2+q^{48} z-q^{46} z^2-q^{46} z+q^{44} z^2-q^{42} z+q^{40}}{}
\\\frac{
   z^2+q^{40} z-q^{38} z^2-q^{36}+q^{34} z+q^{34}-q^{32} z+q^{30}-q^{28} z-q^{28}+q^{26} z+q^{26}-q^{20}
   z-q^{20}+q^{12}-q^6+1)}{q^{10} \qn(3) \qn(5) \qn(30)^2 (q^2 z-1)
   (q^{12} z-1) (q^{20} z-1) (q^{30} z-1)} \cupcap[0.35]\\
 -\frac{\qn(15) (z-1) (q^{10} z-1) (q^{52} z+q^{50} z+q^{48} z+q^{44} z-q^{42} z-q^{38}
   z+q^{36} z-q^{34} z-q^{34}-q^{32}-q^{30} z-q^{30}+q^{28} z-q^{24}+q^{22} z+q^{22}+q^{20} z+q^{18}
   z+q^{18}-q^{16}+q^{14}+q^{10}-q^8-q^4-q^2-1)}{q^{14} \qn(3) \qn(30) (q^2 z-1)
   (q^{12} z-1) (q^{20} z-1)} \vertical[0.35]\\
 -\frac{\qn(15) (z-1) (q^{46} z+q^{44} z+q^{42} z+q^{36} z-q^{34} z-q^{34}-q^{32}-q^{30}
   z-q^{30}+q^{28} z-q^{26} z-q^{26}+q^{24}-q^{22} z+q^{20} z+q^{20}-q^{18}+q^{16} z+q^{16}+q^{14} z+q^{12}
   z+q^{12}-q^{10}-q^4-q^2-1)}{q^{16} \qn(3) \qn(30) (q^2 z-1) (q^{12}
   z-1)} \horizontal[0.35]\\
+ \frac{q^2 \qn(5) (z-1) (q^{10} z-1)}{\qn(6) \qn(10) (q^2 z-1)
   (q^{12} z-1)} \mysquare[0.35]
+ \frac{(q^2-1)^3 \qn(6) \qn(10) \qn(15) z (q^{30} z+1)}{q \qn(30)
   (q^2 z-1) (q^{12} z-1) (q^{20} z-1)} \left(\leftid[0.35]+\rightid[0.35]\right)\\
+ \frac{\qn(6) \qn(10) \qn(15) (z-1) (q^{46} z^2+q^{44} z^2-q^{40} z^2-q^{38} z^2-q^{36}
   z^2-q^{36} z-q^{34} z+q^{32} z^2+q^{30} z^2+q^{30} z+q^{16} z+q^{16}+q^{14}-q^{12} z-q^{10}
   z-q^{10}-q^8-q^6+q^2+1)}{q^6 \qn(2) \qn(3) \qn(5) \qn(30) (q^2 z-1)
   (q^{12} z-1) (q^{20} z-1)} \left(\diagida[0.35]+\diagidb[0.35]\right)\\
+ \kappa^{-1}\frac{(q^2-1)^3 q^{14}\qn(6) \qn(10) \qn(15) (z-1) z}{\qn(30) (q^2 z-1)
   (q^{12} z-1) (q^{20} z-1)} \left(\Yaa[0.35]
+\Yba[0.35]\right)
+ \kappa\frac{(q^2-1) q^{16} (z-1) z}{(q^2 z-1) (q^{12} z-1) (q^{20} z-1)} \left(\Yab[0.35]
+\Ybb[0.35]\right)\\
 -\kappa^{-2}\frac{(q^2-1)^5 q^{27} \qn(6)^2 \qn(10)^2 \qn(15)^2 (z-1) z (z+1)}{\qn(30)^2
   (q^2 z-1) (q^{12} z-1) (q^{20} z-1) (q^{30} z-1)} \cupa[0.35]
 -\kappa^2\frac{(q^2-1) q^{31} (z-1) z (z+1)}{(q^2 z-1) (q^{12} z-1) (q^{20} z-1)
   (q^{30} z-1)} \capa[0.35]\\
+ \frac{\qn(6) \qn(10) \qn(15) (q^{80} z^2+q^{78} z^4-q^{78} z^3+q^{76} z^4-q^{76}
   z^3-q^{76} z^2-q^{74} z^2-q^{72} z^4+q^{72} z^3-q^{70} z^4+q^{70} z^3-q^{68} z^4+q^{68} z^2-q^{66} z^3+2 q^{66}
   z^2+q^{64} z^4-q^{64} z^3+q^{64} z^2+q^{62} z^4-q^{62} z^2}{}
\\\frac{-q^{60} z^2+q^{48} z^2-q^{48} z-q^{46} z^3+2 q^{46}
   z^2-q^{46} z-q^{44} z^3+q^{44} z^2-q^{42} z^2+q^{42} z+q^{40} z^3-2 q^{40} z^2+q^{40} z+q^{38} z^3-q^{38}
   z^2+q^{36} z^2-q^{36} z-q^{34} z^3+2 q^{34} z^2-q^{34} z-q^{32} z^3+q^{32} z^2}{}
\\\frac{
-q^{20} z^2-q^{18}   z^2+q^{18}+q^{16} z^2-q^{16} z+q^{16}
+2 q^{14} z^2-q^{14} z+q^{12} z^2-q^{12}+q^{10} z-q^{10}+q^8 z-q^8-q^6
   z^2-q^4 z^2-q^4 z+q^4-q^2 z+q^2+z^2)}{q^8 \qn(2) \qn(3) \qn(5) \qn(30) (q^2
   z-1) (q^{12} z-1) (q^{20} z-1) (q^{30} z-1)} \myempty[0.35] \\
\end{gather*}%
}%
\end{landscape}

\section{Rational limit}\label{app:rat}
The rational $R$-matrix is obtained from the trigonometric one in the correlated limit $q,z\to 1$:
\[
q=e^{\epsilon\hbar/2}
\qquad
z=e^{\epsilon x}
\qquad
\epsilon\to 0
\]
Representation-theoretically, it corresponds to the limit from the quantized loop algebra $\Uq(\e[z^\pm])$ (i.e., the quotient
of $\Uq(\ee)$ by $\prod_{i=0}^8 k_i^{n_i}=1$) to the Yangian $\mathcal Y(\e)\cong \Uq(\e[z])$.
In this appendix, we show briefly how to recover the results of \cite{Koepsell-Ye8}.

The main simplification in the diagrammatic calculus (see the related discussion at the start of \S\ref{sec:R0})
is that undercrossings and overcrossings become indistinguishable:
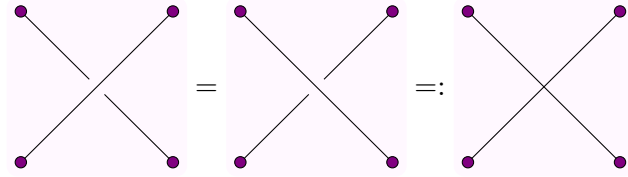
\[
\crossing=\invcrossing=:\flatcrossing
\]
and both become equal to the naive permutation of tensors of $V\otimes V$. The reader is invited
to write the simplified relations that
diagrams satisfy in this limit; we only provide one example:
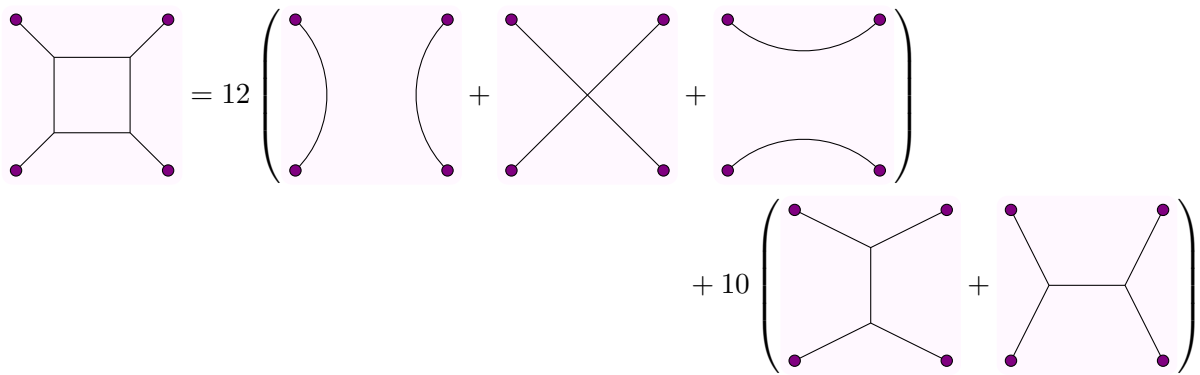
\begin{multline*}
\mysquare=12\left(\full+\flatcrossing+\cupcap\right)
\\
+10\left(\vertical+\horizontal\right)
\end{multline*}

The expression of the $R$-matrix in \S\ref{app:main} becomes in the rational limit:
\begin{align*}
\check R_{\mathrm{rat}}(x)
&=\frac{\hbar-x}{\hbar+x}\full
-\frac{x (5 \hbar+x) (9 \hbar+x) (16 \hbar+x)}{(\hbar+x) (6 \hbar+x) (10 \hbar+x) (15 \hbar+x)}\cupcap
\\
&-\frac{x (5 \hbar+x) (44 \hbar+5 x)}{6 (\hbar+x) (6 \hbar+x) (10 \hbar+x)}\vertical
-\frac{x (31 \hbar+5 x)}{6 (\hbar+x) (6 \hbar+x)}\horizontal
\\
&+\frac{x (5 \hbar+x)}{12 (\hbar+x) (6 \hbar+x)}\mysquare
\\
&+\frac{60 \hbar^3}{(\hbar+x) (6 \hbar+x) (10 \hbar+x)}\left(\leftid+\rightid\right)
\\
&+\frac{x (4 \hbar+x) (11 \hbar+x)}{(\hbar+x) (6 \hbar+x) (10 \hbar+x)}\left(\diagida+\diagidb\right)
\\
&+\frac{\sqrt{30} \hbar^2 x}{(\hbar+x) (6 \hbar+x) (10 \hbar+x)}\left(\Yaa+\Yba+\Yab+\Ybb\right)
\\
&-\frac{60 \hbar^3 x}{(\hbar+x) (6 \hbar+x) (10 \hbar+x) (15 \hbar+x)}\left(\cupa+\capa\right)
\\
&+\frac{
60 \hbar^3 (30 \hbar+x)- (\hbar-x) (6 \hbar+x)(10 \hbar+x) (15 \hbar+x)
}{(\hbar+x) (6 \hbar+x) (10 \hbar+x) (15 \hbar+x)}
\myempty
\end{align*}

The alternative basis of \S\ref{sec:R0} no longer makes sense because $\crossing[0.5]$ and $\invcrossing[0.5]$ are degenerate; we can however use
an intermediate basis, in which the $R$-matrix becomes
\begin{align*}
\check R_{\mathrm{rat}}(x)=&\frac{1}{(\hbar+x)(6\hbar+x)}
\left(
6 \hbar^2\full
+
\frac{6 \hbar^2 x (5 \hbar+x)}{(10 \hbar+x) (15 \hbar+x)}\cupcap
\right.\\
&+\left.
\frac{\hbar x (5 \hbar+x)}{10 \hbar+x}\vertical
-\hbar x\horizontal
+x (5 \hbar+x)\flatcrossing\right)+\cdots
\end{align*}
the other terms (involving empty vertices) remaining the same.

It is not hard to see that the latter result coincides with that of 
\cite{Koepsell-Ye8} (see the equation between (3.9) and (3.10)), being careful that the expression there
is $R(w)=P\check R_{\textrm{rat}}(x=-w)$ where $P$ is permutation of factors, on condition that the following
normalizing factor is used:
\[
f(w)=\frac{w(w-5\hbar)(w-9\hbar)}{(w-\hbar)(w-6\hbar)(w-10\hbar)}
\]
with the substitution $i\hbar\to\hbar$ and with the free parameter $\alpha$ (which is related to our
parameter $\kappa$) given the value
\[
\alpha=\frac{1}{2\sqrt{15}}
\]

\gdef\MRshorten#1 #2MRend{#1}%
\gdef\MRfirsttwo#1#2{\if#1M%
MR\else MR#1#2\fi}
\def\MRfix#1{\MRshorten\MRfirsttwo#1 MRend}
\renewcommand\MR[1]{\relax\ifhmode\unskip\spacefactor3000 \space\fi
\MRhref{\MRfix{#1}}{{\scriptsize \MRfix{#1}}}}
\renewcommand{\MRhref}[2]{%
\href{http://www.ams.org/mathscinet-getitem?mr=#1}{#2}}
\bibliographystyle{amsalphahyper}
\bibliography{biblio}

\end{document}